\documentclass{llncs}

\usepackage[linesnumbered,ruled,titlenumbered,vlined]{algorithm2e}
\usepackage{amsmath}
\usepackage{amssymb}
\usepackage{color}
\usepackage{mathtools}
\usepackage{hyperref}
\usepackage{mathrsfs}
\usepackage[misc,geometry]{ifsym}  \def\longversion{1}

\renewenvironment{proof}{{\noindent \itshape Proof.}}{\qed\vspace{\baselineskip}}

\newcommand{\F}{\mathbb{F}}
\newcommand{\Fq}{\mathbb{F}_{q}}
\newcommand{\Fqm}{\mathbb{F}_{q^m}}

\newcommand{\Fqmu}{\mathbb{F}_{q^{mu}}}

\newcommand{\pk}{\mathbf{pk}}
\newcommand{\rk}[1][q]{\mathbf{rank}_{#1}}
\newcommand{\sk}{\mathbf{sk}}
\newcommand{\tr}{{\rm Tr}}
\newcommand{\Tr}[1][m]{{\rm Tr}_{q^{#1u}/q^{#1}}}
\newcommand{\Prob}{\mathbb{P}}
\newcommand{\gbinom}[2]{\begin{bmatrix}#1 \\ #2
  \end{bmatrix}}
\newcommand{\Stab}{{\rm Stab}}

\newcommand{\img}{{\rm Im}}
\renewcommand{\ker}{{\rm Ker}}

\newcommand{\word}[1]{\ensuremath{\boldsymbol{#1}}}
\newcommand{\bfa}{\word{a}}
\newcommand{\bfc}{\word{c}}
\newcommand{\bfe}{\word{e}}
\newcommand{\bfg}{\word{g}}
\newcommand{\bfk}{\word{k}}
\newcommand{\bfm}{\word{m}}
\newcommand{\bfs}{\word{s}}
\newcommand{\bfu}{\word{u}}

\newcommand{\bfx}{\word{x}}
\newcommand{\bfy}{\word{y}}
\newcommand{\bfz}{\word{z}}

\newcommand{\mat}[1]{\ensuremath{\boldsymbol{#1}}}
\newcommand{\bfA}{\mat{A}}
\newcommand{\bfB}{\mat{B}}
\newcommand{\bfC}{\mat{C}}
\newcommand{\bfE}{\mat{E}}
\newcommand{\bfG}{\mat{G}}
\newcommand{\bfH}{\mat{H}}
\newcommand{\bfM}{\mat{M}}
\newcommand{\bfP}{\mat{P}}
\newcommand{\bfS}{\mat{S}}
\newcommand{\bfT}{\mat{T}}
\newcommand{\bfY}{\mat{Y}}

\newcommand{\bigzero}{\large\mat{0}}

\newcommand{\kpub}{\word{k}_{pub}}
\newcommand{\kpriv}{\bfk_{{\rm priv}}}
\newcommand{\Ksec}{K_{\textrm{sec}}}

\newcommand{\sample}{\xleftarrow{\$}}

\newcommand{\code}[1]{\mathscr{#1}}
\newcommand{\BB}{\code{B}}
\newcommand{\CC}{\code{C}}
\newcommand{\EE}{\code{E}}
\newcommand{\FF}{\code{F}}
\newcommand{\GG}{\code{G}}
\newcommand{\MM}{\code{M}}
\newcommand{\TT}{\code{T}}
\newcommand{\Cpub}{\CC_{\rm pub}}
\newcommand{\Gab}[2]{\code{G}_{#1}(#2)}
\newcommand{\Frob}[2]{{#1}^{[#2]}}
\newcommand{\ext}[2]{\textrm{Ext}_{#2}(#1)}
\newcommand{\Supp}{\textrm{Supp}}
\newcommand{\RSupp}{\textrm{RowSupp}}

\newcommand{\Mspace}[2]{\mathcal{M}_{#1}(#2)}
\newcommand{\GL}[2]{\textrm{GL}_{#1}(#2)}

\newcommand{\qpoly}[1]{\mathcal{L}_{#1}}
\newcommand\smvee{\hbox{$\scriptscriptstyle\vee$}}
\newcommand{\adj}[1]{{#1}^{\raise0.9ex\smvee}}
\newcommand{\wadj}[1]{{#1}^{\smvee}} \newcommand{\qdeg}{\deg_q}

\newcommand{\eqdef}{\stackrel{\textrm{def}}{=}}
\renewcommand{\span}{\mathbf{Span}}
\renewcommand{\leq}{\leqslant}
\renewcommand{\le}{\leqslant}

\renewcommand{\ge}{\geqslant}
\newcommand{\map}[4]{\left\{
    \begin{array}{ccc}
      #1 & \longrightarrow & #2\\
      #3 & \longmapsto & #4
    \end{array}\right.
}
\newcommand{\ie}{{\em i.e. }}

\newcommand{\LIGA}{{\sc Liga}}
\newcommand{\RAMESSES}{{\sc Ramesses}}

 \author{Maxime Bombar\inst{1,2}\textsuperscript{({\tiny\Letter}) } \and Alain Couvreur \inst{2,1}}

\institute{LIX, CNRS UMR 7161, \'Ecole Polytechnique,\\
  Institut Polytechnique de Paris,\\
1 rue Honor\'e d'Estienne d'Orves\\
91120 {\sc Palaiseau Cedex} \and
Inria\\
\email{\{maxime.bombar, alain.couvreur\}@inria.fr}
}
 
\begin{document}

\title{Decoding Supercodes of Gabidulin Codes and
  Applications to Cryptanalysis}

\maketitle

\begin{abstract}
  This article discusses the decoding of Gabidulin codes and shows how
  to extend the usual decoder to any supercode of a Gabidulin code at
  the cost of a significant decrease of the decoding radius. Using
  this decoder, we provide polynomial time attacks on the rank metric
  encryption schemes \RAMESSES{} and \LIGA{}.

\end{abstract}

 \section*{Introduction}
It is well--known that error correcting codes lie among the possible
candidates for post quantum cryptographic primitives.  For codes
endowed with the Hamming metric the story begins in 1978 with
McEliece's proposal \cite{M78}.  The security of this scheme relies on
two difficult problems: the hardness of distinguishing classical Goppa
codes from arbitrary codes and the hardness of the syndrome decoding
problem.  To instantiate McEliece scheme, the only requirement is to
have a family of codes whose structure can be hidden and benefiting
from an efficient decoder. In particular, this paradigm does not
require the use of codes endowed with the Hamming metric.
Hence other metrics may be considered such as the rank metric,
as proposed by Gabidulin, Paramonov and Tretjakov in \cite{GPT91}.

Besides McEliece's paradigm, another approach to perform encryption
with error correcting codes consists in using codes whose structure is
no longer hidden but where encryption is performed so that decryption
without the knowledge of the secret key would require to decode the
public code far beyond the decoding radius. This principle has been
first instantiated in Hamming metric by Augot and Finiasz in
\cite{AF03} using Reed--Solomon codes. Later, a rank metric
counterpart is designed by Faure and Loidreau in \cite{FL05}.  Both
proposals have been subject to attacks, by Coron \cite{C03} for the
Hamming metric proposal and by Gaborit, Otmani and Tal\'e--Kalachi
\cite{GOT18} for the rank metric one. More recently, two independent
and different repairs of Faure--Loidreau scheme resisting to the
attack of Gaborit, Otmani and Tal\'e--Kalachi appeared. The first one,
due to Renner, Puchinger and Wachter--Zeh and is called \LIGA{}
\cite{WPR18,RPW20}. The second one, due to Lavauzelle, Loidreau and
Pham is called \RAMESSES{} \cite{LLP20}.

\if\longversion1
\subsection*{Our contribution}
\else
\subsubsection{Our contribution}
\fi
In the present article, we show how to
extend the decoding of a Gabidulin code to a supercode at the cost of
a significant decrease of the decoding radius. With this decoder in hand
we perform  a polynomial time message
recovery attack on \RAMESSES{} and \LIGA{}.

 \section{Notation and Prerequisites}\label{sec:prerequisites}
In this article, we work over a finite field $\Fq$ and we
frequently consider two nested extensions denoted $\Fqm$ and $\Fqmu$.
Rank metric codes will be subspaces $\CC \subseteq \Fqm^n$,
the code length and dimension are respectively denoted as
$n$ and $k$.

Vectors are represented by lower case bold face letters such as
$\bfa, \bfc, \bfe$ and matrices by upper case letters
$\bfA, \bfG, \bfM$. The space of $m\times n$ matrices with entries in
a field $\mathbb K$ is denoted by $\Mspace{m,n}{\mathbb K}$. When
$m = n$ we denote it by $\Mspace{n}{\mathbb K}$ and the group of
non-singular matrices is denoted by $\GL{n}{\mathbb K}$.  Given a
matrix $\bfM \in \Mspace{m,n}{\Fq}$ its rank is denoted by
$\rk(\bfM)$.  Similarly, given a vector $\bfc \in \Fqm^n$, the {\em
  $\Fq$--rank} or {\em rank} of $\bfc$ is defined as the dimension of
the subspace of $\Fqm$ spanned by the entries of $\bfc$. Namely,
\[
  \rk (\bfc) \eqdef \dim_{\Fq} \left(\span_{\Fq} \{c_1, \dots, c_n\}\right).
\]

We will consider two notions of \emph{support} in the rank
metric. Inspired by the Hamming metric, the most natural one which we
will denote by the \emph{column support} of a vector $\bfc\in\Fqm^{n}$
is the linear subspace spanned by its coordinates:
\[
  \Supp(\bfc) \eqdef \span_{\Fq}\{c_{1},\dots, c_{n}\}.
\]
\if\longversion1
\noindent But we can define another notion, namely the \emph{row
  support}. Let $\BB$ be a basis of the extension field
$\Fqm/\Fq$. Then, we define the extension of $\bfc$ with respect to
$\BB$ as the matrix $ \ext{\bfc}{\BB}\in \Mspace{m,n}{\Fq}$ whose
columns are the entries of $\bfc$ represented in the basis $\BB$. The
row space of $\ext{\bfc}{\BB}$ with respect to $\BB$ will be called the
{\em row support} of $\bfc$, \ie
\[
  \RSupp(\bfc) \eqdef \{\bfx\ext{\bfc}{\BB} \mid \bfx \in
  \Fq^{m}\} \subset \Fq^{n}.
\]
Notice that the above definition does not depend on the choice
of the basis $\BB$.
\fi

Given $\bfc = (c_1, \dots, c_n) \in \Fqm^n$, and $j \in \{0, \dots, m-1\}$,
we denote
\[
  \Frob{\bfc}{j} \eqdef (c_1^{q^j}, \dots, c_n^{q^j}).
\]
Similarly, for a code $\CC \subseteq \Fqm^n$, we denote
\if\longversion1
\[
  \else
  \(
  \fi
  \Frob{\CC}{j} \eqdef \left\{\Frob{\bfc}{j} ~\big|~ \bfc \in \CC\right\}.
  \if\longversion1
\]
\else
\)
\fi
Let $n \leq m$, $k \leq n$ and $\bfg \in \Fqm^n$ with $\rk (\bfg) = n$,
the {\em Gabidulin code of dimension $k$ supported by $\bfg$} is defined as
\[
  \Gab{k}{\bfg} \eqdef \span_{\Fqm} \left\{\Frob{\bfg}{i} ~\big|~ 0 \leq i \leq
    k-1 \right\}.
\]

A {\em $q$--polynomial} is a polynomial $P \in \Fqm [X]$ whose monomials
are only $q$--th powers of $X$, {\em i.e.} a polynomial of the form
\if\longversion1
\[
  \else
  \(
  \fi
  P(X) = p_0 X + p_1 X^q + \cdots + p_r X^{q^r}.
  \if\longversion1
\]
\else
\)
\fi
Assuming that $p_r \neq 0$ then the integer $r$ is called the {\em
  $q$--degree} of $P$.  Such a polynomial induces an
$\Fq$--linear map $P : \Fqm \rightarrow \Fqm$ and we call the {\em
  rank} of the $q$--polynomial, the rank of the induced
map.
A well--known fact on $q$--polynomials is that the  $\Fq$--dimension
of the kernel of the induced endomorphism is bounded from
above by their $q$-degree.
Conversely, any $\Fq$--linear endomorphism of $\Fqm$ is uniquely
represented by a $q$--polynomial of degree $< m$.  Denote by
$\qpoly{}$ the space of $q$--polynomials, this space equipped with the
composition law is a non commutative ring which is left and right
Euclidean \cite[\S~1.6]{G96b} and the two--sided ideal
$(X^{q^m} - X)$ is the kernel of the canonical map
\[
  \qpoly{} \longrightarrow \textrm{Hom}_{\Fq}(\Fqm, \Fqm)
\]
inducing an isomorphism :
\if\longversion1
\[
\else
\(
\fi
\qpoly{} / (X^{q^m}-X) \simeq \textrm{Hom}_{\Fq}(\Fqm, \Fqm).
\if\longversion1
\]
\else
\)
\fi
Finally, given a positive integer $k < m$, we denote by $\qpoly{< k}$
(resp. $\qpoly{\leq k}$) the space of $q$--polynomials of $q$--degree
less than (resp. less than or equal to) $k$.
The Gabidulin code $\Gab{k}{\bfg}$ is canonically isomorphic to
$\qpoly{< k}$ under the isomorphism:
\[
  \map{\qpoly{< k}}{\Gab{k}{\bfg}}{P}{(P(g_1), \dots, P(g_n)).}
\]
The above map is actually an isometry: it is rank preserving.  In this
article, we will extensively use this isometry and Gabidulin codes
will be represented either as an evaluation code or as a space of
$q$--polynomials of bounded degree $\qpoly{<k}$, when one
representation is more suitable than the other.
\if\longversion1
In particular, given two $\Fqm$--linear subspaces $\code{A}, \code{B}$
of $\qpoly{}/(X^{q^m}-X)$ we define their {\em composition} as
\[
  \code{A}\circ\code{B} \eqdef \span_{\Fqm} \{P \circ Q ~|~
  P\in \code{A},\ Q \in \code{B}\}.
\]
This definition may be regarded as a rank metric analogue of
the so--called Schur product of codes in Hamming metric.
\fi

Another notion which is very useful in the sequel is the notion of {\em
  adjoint} of a class of $q$--polynomial $P = \sum_{i=0}^{m-1} p_i X^{q^i}$
in $\qpoly{}/(X^{q^m}-X)$,
which is defined as
\begin{equation}\label{eq:adjunction}
  \adj{P}(X) \eqdef \sum_{i=0}^{m-1}
  X^{q^{m-i}} p_i = \sum_{i=0}^{m-1} p_i^{q^{m-i}} X^{q^{m-i}}.
\end{equation}
Regarding $P$ as an $\Fq$--linear endomorphism of $\Fqm$, the notion
of adjoint is nothing but the usual notion of {\em adjoint} or {\em
  transposed} endomorphism with respect to the inner product
\if\longversion1
\[
  \else
  \( \fi \map{\Fqm \times \Fqm}{\Fq}{(x,y)}{\tr_{\Fqm/\Fq} (xy)}
  \if\longversion1 .\] \else \) \fi
\if\longversion0 (see
\cite[Section~2.4]{ACLN20}) \fi 
\if\longversion1 Details are given in a more general context in
\cite[Section~4.2]{ACLN20}, we give them here in the context of
$q$--polynomials for the sake of self--containedness.

With respect to this bilinear form, the multiplication map by
a scalar $x \mapsto ax$ for $a \in \Fqm^{\times}$ is {\em symmetric}
(or {\em self-adjoint}) since,
\[
  \langle x, ay \rangle = \tr_{\Fqm/\Fq} (xay) = \tr_{\Fqm/\Fq}(axy) =
  \langle ax, y \rangle.
\]
Next, the Frobenius endomorphism is {\em orthogonal}, that is to say
its adjoint $\adj{(X^q)}$ is its inverse $X^{q^{m-1}}$ because:
\[
  \langle x, y^q \rangle = \tr_{\Fqm/\Fq}(xy^q) =
  \tr_{\Fqm/\Fq}\left({(xy^q)}^{q^{m-1}}\right) = \tr_{\Fqm/\Fq}(x^{q^{m-1}}y)=
  \langle x^{q^{m-1}}, y \rangle.
\]
Finally, using that adjunction is anticommutative, \ie for any
$\Fq$--endomorphisms $f, g$ of $\Fqm$ we have
$\adj{(fg)} = \adj{g}\adj{f}$ we can prove that
Definition~\ref{eq:adjunction} coincides with that of adjoint
endomorphism.
\fi
In particular, for any
$P \in \qpoly{}/(X^{q^m}-X)$, we have $\rk{(P)} = \rk{(\adj{P})}$.

 \section{Two Rank Metric Proposals with Short Keys}
\subsection{\LIGA{} encryption
  scheme}\label{sec:FL_and_LIGA}

In this section, we recall Faure--Loidreau cryptosystem \cite{FL05}
and the repaired version \cite{WPR18} recently extended to
proposal \LIGA{}~\cite{RPW20}.

\begin{description}
  \item{\bf Parameters}
Let $q, m, n, k, u, w$ be positive integers such that $q$ is a prime
power, $u < k < n$ and
\[
	n-k > w > \lfloor \tfrac{n-k}{2} \rfloor.
\]
In the following, we consider the three finite fields
\if\longversion1
  \[
\else
  \(
\fi
   \Fq \subseteq \Fqm \subseteq \Fqmu.
\if\longversion1
  \]
\else
  \)
\fi

Let $t_{pub} \eqdef \lfloor \tfrac{n-k-w}{2} \rfloor$ be the public
$\Fq$--rank of the error in the ciphertext, and let $\bfG$ be a
generator matrix of a public Gabidulin code of length $n$ and
dimension $k$ over $\Fqm$.

\item{\bf Key generation.} Alice picks uniformly at random a
vector $\bfx\in \Fqmu^{k}$ whose last $u$ entries form an
$\Fqm$--basis of $\Fqmu$, and a vector
$\bfz\in\Fqmu^{n}$ of $\Fq-$rank $w$. In order to do that, she chooses
a full-rank vector $\bfs\in\Fqmu^{w}$ and a non-singular matrix
$\bfP \in \GL{n}{\Fq}$ and sets
$$\bfz = (\bfs\mid \mathbf{0})\cdot\bfP^{-1}.$$

The \emph{private key} is then $(\bfx, \bfz, \bfP)$ and the \emph{public key} is the \emph{vector}
\[
	\kpub \eqdef \bfx \cdot \bfG + \bfz \in \Fqmu^{n}.
\]

The key generation is summarised 
by Algorithm~\ref{keygenFL}.
\begin{algorithm}
	\DontPrintSemicolon
	\KwIn{Parameters $q, \bfG, m, n, k, u, w$.}
	\KwOut{Private key $\sk$, and public key $\pk$}
	$\bfx \sample \{\bfa\in\Fqmu^{k}\mid \dim(\span_{\Fqm} \{ a_{k-u+1}, \dots, a_{k}\}) = u \}$\;
	$\bfs \sample \{\bfa\in\Fqmu^{w}\mid \rk(\bfa) = w \}$\;
	$\bfP \sample \GL{n}{\Fq}$\;
	$\bfz \leftarrow (\bfs\mid \mathbf{0})\cdot\bfP^{-1}$\;
	$\kpub \leftarrow \bfx \cdot \bfG + \bfz$\;
	$\sk \leftarrow (\bfx, \bfz, \bfP)$\;
	$\pk \leftarrow \kpub$\;
	\KwRet{$(\sk, \pk)$}
	\caption{Original Faure--Loidreau Key Generation\label{keygenFL}}
\end{algorithm}
 
\item{\bf Encryption.} Let
$\bfm = (m_{1}, \dots, m_{k-u}, 0, \dots, 0) \in \Fqm^{k}$ be the
plaintext. Note that the last $u$ entries are chosen to be zero in
order to be able to decrypt.  The encryption of $\bfm$ works as
follows:

\begin{enumerate}
	\item Pick $\alpha\in\Fqmu$ at random.
	\item Pick $\bfe\in\Fqm^{n}$ such that $\rk(\bfe) \le t_{pub}$ at random.
\end{enumerate}

The ciphertext is then $\bfc\in\Fqm^{n}$:
\[
	\bfc = \bfm\cdot \bfG + \Tr(\alpha\kpub) + \bfe.
\]

As shown in (\ref{eq:action_of_kpub}) below, the public key acts on
the one hand as a one-time pad on the message $\bfm$, and on the other
hand adds a random error of large weight. The ciphertext can indeed be
seen as a codeword of the Gabidulin code corrupted by a two-part
errors formed by the private key $\bfz$ and the random error vector
$\bfe$:

\begin{equation}\label{eq:action_of_kpub}
	\bfc = (\bfm + \Tr(\alpha\bfx))\cdot \bfG + (\Tr(\alpha\bfz) + \bfe).
\end{equation}

With very high probability, the error in the \emph{ciphertext} is of rank-weight $w + t_{pub}$. See \cite{RPW20} for a detailed discussion about the parameters in order to avoid so-called weak keys.

\item{\bf Decryption.} The \emph{receiver} first computes
\[
	\bfc\cdot \bfP = (\bfm + \Tr(\alpha\bfx))\cdot \bfG\bfP + (\Tr(\alpha\bfs)\mid 0) + \bfe\bfP.
\]
whose last $n-w$ entries are given by
\[
	\bfc' = (\bfm + \Tr(\alpha\bfx))\cdot \bfG' + \bfe',
\]
where $\bfG'$ is the generator matrix of a Gabidulin code of length
$n-w$ and dimension $k$ and $\bfe'$ is an error vector of rank-weight
at most $t_{pub} = \lfloor \tfrac{n-w-k}{2}\rfloor$.

By decoding in this new Gabidulin code, the \emph{receiver} 
obtains the vector
\[
	\bfm' = \bfm + \Tr(\alpha\bfx).
\]
Since by construction $\bfm$ is chosen such that its last $u$
components are $0$ and the last $u$ components of $\bfx$ form a basis
of $\Fqmu/\Fqm$, the \emph{receiver} can 
\if\longversion1
compute
\[
	\alpha = \sum_{i=k-u+1}^{k}m'_{i}x_{i}^{\ast}
\]
where $(x_{k-u+1}^{\ast}, \dots, x_{k}^{\ast})$ is the dual basis of
$(x_{k-u+1}, \dots, x_{k})$ with respect to the on degenerate bilinear
form $(x,y) \mapsto \Tr(xy)$.
Knowing both $\alpha$ and $\bfx$, the \emph{receiver} can finally recover the plaintext $\bfm$.
\else
deduce the plaintext $\bfm$ from the knowledge of $\bfm'$ and $\bfx$.
\fi
This encryption scheme has no decryption failure.
\end{description}

\subsubsection{A Key Recovery Attack.} In \cite{GOT18},
Gaborit, Otmani and
  Talé--Kalachi showed that a valid private key for this system could
be efficiently computed from $\kpub$, and later in \cite{WPR18} and
\cite{RPW20}, Renner, Puchinger and Wachter--Zeh introduced a
coding-theoretic interpretation of the public key as a corrupted
codeword of an $u$--interleaved Gabidulin code. They derived an equally
powerful key recovery attack, and proved that the failure conditions
of both attacks were equivalent.

Based on this interpretation, Renner et. al. proposed to change the key
generation algorithm to resist previous attacks. More precisely, they
proved that if $\zeta$ denotes the dimension of the
$\Fqm$--support of $\bfz$, all then known attacks were inefficient when
$\zeta < \frac{w}{n-k-w}$. The new key generation can be summarised in
Algorithm~\ref{keygenLIGA}.

\begin{algorithm}
	\DontPrintSemicolon
	\KwIn{Parameters $q, \bfG, m, n, k, u, w, \zeta$.}
	\KwOut{Private key $\sk$, and public key $\pk$}
	$\mathbf{\gamma} \sample \{\bfa\in\Fqmu^{u}\mid \rk[q^{m}](\bfa) = u\}$\;
	$\bfx \sample \{\bfa\in\Fqmu^{k}\mid \dim(\span_{\Fqm}\{ a_{k-u+1}, \dots, a_{k}\}) = u \}$\;
	$\mathbf{\mathcal{A}} \sample \{\text{subspaces } \mathcal{U}\subseteq \Fqm^{w} \mid \dim\mathcal{U} = \zeta, \ \mathcal{U}\text{ has a basis of full-$\Fq$-rank elements}\}$\;
	$\begin{pmatrix}
			\bfs_{1} \\
			\vdots   \\
			\bfs_{u}
		\end{pmatrix} \sample \left\{\begin{pmatrix}
			\bfs'_{1} \\
			\vdots    \\
			\bfs'_{u}
		\end{pmatrix} \mid \langle \bfs'_{1}, \dots, \bfs'_{u} \rangle_{\Fqm} = \mathbf{\mathcal{A}},\ \rk(\bfs'_{i}) = w \ \forall i \right\}$\;
	$\bfs \leftarrow \sum_{i=1}^{u} \bfs_{i}\mathbf{\gamma_{i}^{\ast}}$\;
	$\bfP \sample \GL{n}{\Fq}$\;
	$\bfz \leftarrow (\bfs\mid \mathbf{0})\cdot\bfP^{-1}$\;
	$\kpub \leftarrow \bfx\bfG + \bfz$\;
	$\sk \leftarrow (\bfx, \bfz, \bfP)$\;
	$\pk \leftarrow \kpub$\;
	\KwRet{$(\sk, \pk)$}
	\caption{\LIGA{} Key Generation\label{keygenLIGA}}
\end{algorithm}

 \subsection{Ramesses}\label{subsec:ramesses}
In this section, we present the proposal \RAMESSES{} \cite{LLP20}
which is another repair of the Faure--Loidreau scheme. We chose to
describe the scheme in a rather different manner which turns out to be
completely equivalent to the original proposal. The connection
between this point of view and that of the original article is
detailed in Appendix~\ref{app:ramesses}.  Our presentation rests only
on $q$--polynomials. As explained in Section~\ref{sec:prerequisites},
the space $\qpoly{< k}$ will be regarded as a Gabidulin code of
dimension $k$. We also fix an $\Fq$--basis $\BB$ of $\Fqm$, which
permits to have an $m \times m$ matrix representation of
$q$--polynomials (modulo $(X^{q^m}-X)$) and conversely provides a
description of any $m \times m$ matrix with entries in $\Fq$ as a
$q$--polynomial of $q$--degree less than $m$.

\begin{description}
\item[{\bf Parameters}] The public parameters are integers
  $1 \leq w, k, \ell, t \leq m$ and should satisfy
  \begin{equation}\label{eq:decoding_condition}
    t \leq \frac{n-k-\ell - w}{2}\cdot
  \end{equation}
\item[{\bf Key generation}] Alice picks a uniformly random
  $q$--polynomial $\Ksec$ of rank $w$. The public key is the
  affine space:
  \[\Cpub \eqdef \Ksec + \qpoly{<k}.\]
\item[{\bf Encryption}] The plaintext $\bfm$ is a $t$--dimensional
  $\Fq$--subspace of $\Fqm$.  It is encrypted as follows:
  \begin{itemize}
  \item Pick a uniformly random $T \in \qpoly{}$ of $q$--degree $\ell$
  \item Pick a uniformly random $E \in \qpoly{<m}$ whose matrix
    representation admits $\bfm$ as its row space, equivalently
    $E$ is such that 
    $\bfm$ is the image of $\adj{E}$.
  \item Pick a uniformly random $C \in \qpoly{<k}$
  \item Pick a uniformly random $C_0 \in \qpoly{<k}$, yielding
    a uniformly random
    \[C' = C_0 + \Ksec \in \Cpub.\]
  \end{itemize}
  The ciphertext is
  \begin{equation}\label{eq:ramesses_ciphertext}
    Y \eqdef C + C' \circ T + E.
  \end{equation}
  Note that, this cipher text satisfies
  \begin{equation}\label{eq:ramesses_ciphertext_rewritten}
    Y = C_1 + \Ksec \circ T + E,
  \end{equation}
  where
  $C_1 = C + C_0 \circ T \in \qpoly{<k} + \qpoly{<k} \circ T \subseteq \qpoly{<k +
    \ell}$. This $C_1$ is {\em a prioiri} unknown by anyone.
\item[{\bf Decryption}] The owner of $\Ksec$ knows a $q$--polynomial
  $V \in \qpoly{\leq w}$ such that
  $V \circ \Ksec \equiv 0 \mod (X^{q^m}-X)$. Hence she can compute
  \[
    V \circ Y \equiv V \circ C_1 + V \circ E \mod (X^{q^m}- X).
  \]
  Now, $V \circ C_1 \in \qpoly{<k+\ell + w}$, {\em i.e.}
  lies in a Gabidulin code, while $\rk (V \circ E) \leq \rk (E) = t$.
  Hence, thanks to (\ref{eq:decoding_condition}), one can deduce
  $V \circ E$ and as soon as $\rk (V \circ E) = t$, the
  row space of the matrix representation of $E$ is that of $V \circ E$
  which can be recovered.
\end{description}

 \section{Decoding of Gabidulin Codes on the Right}
\label{subsec:right_decoding}

In this section, we assume that $n=m$, \ie
$\bfg \eqdef (g_{1}, \dots, g_{n})$ forms a basis of the extension
field $\Fqm/\Fq$. Let $\CC$ be a Gabidulin code of dimension $k$ and
support $\bfg$. Suppose we receive a vector $\bfy = \bfc + \bfe$ where
$\bfc\in\CC$ and $\bfe$ has rank $t \le
\lfloor\frac{n-k}{2}\rfloor$. There exist three $q$--polynomials
$C \in \qpoly{<k}$ and $Y, E \in \qpoly{<m}$ with $\rk(E) = t$ such that
\begin{equation}\label{eq:decoding_pb}
  Y = C + E
\end{equation}
and the polynomial $Y$ can be deduced from the knowledge of
$\bfy$ and the basis $\bfg$ by interpolation (see for instance
\cite[Chapter~3]{W13}).

\begin{remark}
  Note that the requirement $n=m$ is necessary.  Indeed, if $n < m$
  the choice of the interpolating polynomial $Y\in \qpoly{<m}$ is not
  unique and a wrong choice for $Y$ yields an $E = Y- C$ of too large
  rank. It is not clear to us how to weaken this condition.
\end{remark}

In a nutshell, our approach can be explained as follows. Starting from
the decoding problem~(\ref{eq:decoding_pb}) and applying the
adjunction operator we have to solve the problem
\[
  \adj{Y} = \adj{C} + \adj{E},
\]
where $\rk{\adj{E}} = \rk{E} \leq \lfloor \frac{n-k}{2} \rfloor$
and $\adj{C}$ is contained in a code which is equivalent to a Gabidulin
code. Hence, $\adj{C}$ can be recovered by applying the decoding
algorithm of \cite{L06a}.
\if\longversion1

In what follows, we give a detailed and self--contained presentation
of how to apply and implement this algorithm practically. We believe that this algorithm might be folklore, but
we weren't able to find it in the literature.

Starting from the decoding problem $Y = C+E$, the decoding problem can be
thought as finding the $q$--polynomial $C$, given $Y$. Using the
analogy with Reed--Solomon codes, Loidreau introduced in \cite{L06a} a
Welch--Berlekamp like algorithm to decode Gabidulin codes that
consists in finding the unique $q$--polynomial $V$ of $q$--degree less
than or equal to $t$ such that $V$ vanishes on the column support of
$\bfe$, which is equivalent to $V\circ E=0$, \ie $V$ is a left
annihilator of the error. Using a linearisation technique, this leads
to the resolution of a linear system that can be efficiently solved
provided that $t$ is less than half the minimum distance. It then
suffices to compute a left Euclidean division to recover $C$ and
therefore the codeword $\bfc$.

The core of the algorithm to follow consists in searching a
right--hand side annihilator of $E$ instead of a left--hand side one.  Due
to the non commutativity of the ring $\qpoly{}$, working on the
right--hand side is not directly equivalent to working on the left--hand side.

We begin to state the existence of a right--hand side annihilator.

\begin{proposition} Let $E$ be a $q$--polynomial of rank $t$. Then there
exists a unique monic $q$--polynomial $V$ with $\qdeg(V)\le t$ such that $E\circ V=0$ modulo
$(X^{q^m}-X)$.
\end{proposition}

\begin{proof} Let $Q \eqdef \sum_{i=0}^{t}a_{i}X^{q^{i}}$ be the unique monic $q$--polynomial of $q$--degree less than or equal to
  $t$ that vanishes exactly on $\img(\adj{E})$, \ie $\img(\adj{E}) = \ker Q$. Such a
  $q$--polynomial is guaranteed to exist (see for instance \cite{B15a} or \cite{O33}.) It follows that $\ker(E)=\img(\adj{Q})$.
  Moreover, $$\adj{Q} = \sum_{i=0}^{t}a_{i}^{q^{m-i}}X^{{q^{m-i}}} =
\left(\sum_{i=0}^{t}a_{t-i}^{q^{m-t+i}}X^{q^{i}}\right) \circ
X^{q^{m-t}}.$$

Let $V$ be the leftmost factor of $\adj{Q}$ in the above
decomposition. It is a $q$--polynomial of $q$--degree $t$, and
$E\circ\adj{Q}=0$ leads to $E\circ V\circ X^{q^{m-t}}=0$. Since
$X^{q^{m-t}}$ is invertible in $\qpoly{}/( X^{q^m} - X )$,
we get $E\circ V=0 \mod (X^{q^m}-X)$.
\end{proof}

The goal is to compute this right--hand side annihilator $V$. It
satisfies
\begin{equation}\label{eq:right_wb} Y\circ V = C\circ V + E\circ V \equiv
C\circ V \mod (X^{q^{m}}-X).
\end{equation}
Equation
\eqref{eq:right_wb} leads to a non linear system of $n$ equations
whose variables are the $t+k+1$ unknown coefficients of $C$ and $V$.

\begin{equation}\label{eq:right_wb_nl_system} \left\lbrace
	\begin{array}{l} (Y\circ V)(g_{i}) = C\circ V(g_{i}) \\ \qdeg V
\le t \\ \qdeg C \le k-1.
	\end{array} \right.
\end{equation}

\noindent Due to the non linearity, it is not clear how this can efficiently be
solved. That is why we consider instead the following linearised
system
\begin{equation}\label{eq:right_wb_sl_system} \left\lbrace
	\begin{array}{l} (Y\circ V)(g_{i}) = N(g_{i}) \\ \qdeg V \le t \\
\qdeg N \le k + t - 1,
	\end{array} \right.
\end{equation}
whose unknowns are the
$k+2t+1$ coefficients of $N$ and $V$. The latter is \emph{a priori}
more general than the former. But we can link the set of solutions of
the two systems. This is specified in the following two propositions.

\begin{proposition} Any solution $(V, C)$ of \eqref{eq:right_wb_nl_system}
gives a solution $(V, N=C\circ V)$ of \eqref{eq:right_wb_sl_system}.
\end{proposition}

\begin{proof}
This is the direct analogue of \cite[Proposition~1]{L06a}.
\end{proof}

\begin{proposition} Assume that $E$ is of rank $t\le \lfloor
\frac{n-k}{2}\rfloor$. If $(V, N)$ is a nonzero solution of
\eqref{eq:right_wb_sl_system} then $N = C\circ V$ where $C = Y-E$ is
the interpolating $q$--polynomial of the codeword.
\end{proposition}

\begin{proof} Let $(V, N) \neq (0, 0)$ be a solution of
\eqref{eq:right_wb_sl_system}, and let $C$ be the $q$--polynomial of
$q$--degree strictly less than $k$ that interpolates the codeword. Let
$R \eqdef N - C\circ V$. It is a $q$--polynomial, of $q$--degree at
most $k-1+t$. Assume that $R \neq 0$. Then,
\[
  (Y-C)\circ V = Y\circ V - C\circ V = N - C\circ V \equiv R \mod (X^{q^m}-X)
\]
  \noindent \ie
  \begin{equation}\label{eq:proof_eq_systems} E\circ V \equiv R \mod (X^{q^m}-X).
  \end{equation}
  \noindent Hence, $\rk(R) \le \rk(E) \le t$. Since $R\neq 0$,
  $\qdeg R \ge \dim\ker R$. Therefore, by the rank--nullity theorem,
  \[ n = \dim\ker R + \rk(R) \le \qdeg R + \rk(R) \le k-1+2t \le n-1 <
n
  \] which is a contradiction. Therefore, $R$ must be zero, \ie $N = C\circ V$.
\end{proof}

Thenceforth, whenever $t\le\lfloor\frac{n-k}{2}\rfloor$, any non zero
solution of \eqref{eq:right_wb_sl_system} allows to recover the
codeword by simply computing a right--hand side Euclidean division, which can be
done efficiently (see \cite{O33}). The decoding process boils down to
solving the system of equations
\eqref{eq:right_wb_sl_system}. However, despite the transformation,
the system is only semi-linear over $\Fqm$. To address this issue, we will again use
the adjoint of a (class of) $q$--polynomial. Let
$\wadj{y_{i}}\eqdef \adj{Y}(g_{i})$,
for all $i=1,\dots, n$. Using the anticommutativity of the adjoint
operator, system \eqref{eq:right_wb_sl_system} is equivalent to

\begin{equation}\label{eq:right_wb_linear_system}
\adj{V}(\wadj{y_{i}}) = \adj{N}(g_{i}) \mbox{ for } i=1,\dots,n.
\end{equation}

\noindent which is now an $\Fqm$--linear system of $n$ equations
whose unknowns are the coefficients of $\adj{V}$ and $\adj{N}$, that
are in explicit one-to-one correspondence with the coefficients of $V$
and $N$.

\begin{algorithm}\label{algo:right_welch_berlekamp}
	\caption{Right--hand side variant of Welch--Berlekamp}
\DontPrintSemicolon \KwIn{$q$ a prime power, $k,n,m$ integers,
$\bfg=(g_{1},\dots,g_{n})$ a basis of $\Fqm/\Fq$, $\CC$ a Gabidulin
code of dimension $k$ and support $\bfg$, $t\le
\lfloor\frac{n-k}{2}\rfloor$ an integer, $\bfy\in\Fqm^{n}$.}
\KwOut{$\bfc\in\CC$ such that $\bfy=\bfc+\bfe$ for some
$\bfe\in\Fqm^{n}$ with $\rk(\bfe)\le t$.}  Find $Y$ the
$q$--polynomial of $q$--degree strictly less than $n$ such that
$Y(g_{i}) = y_{i}$\; Compute $\adj{Y}$ and evaluate in $\bfg$ to get
$\wadj{\bfy} \eqdef \adj{Y}(\bfg) \in\Fqm^{n}$\; Find a non zero solution
$(V_{0}, N_{0})$ of the linear system
\eqref{eq:right_wb_linear_system}\; Compute $V \eqdef \adj{V_{0}}$ and $N
\eqdef \adj{N_{0}}$\; Recover $C$ by computing the right--hand side
Euclidean division of $N$ by $V$\; \KwRet{$\bfc \eqdef C(\bfg)$}
\end{algorithm}

An implementation of this algorithm using SageMath v9.2 \cite{sage9.2} can be found on Github: \url{https://github.com/mbombar/Attack_on_LIGA}.

\if\longversion1
\begin{remark}
  This right--hand side algorithm can be generalised in order to
  decode an $u$--interleaved Gabidulin code (see for instance
  \cite{WZ14} for further reference about interleaved Gabidulin codes
  and their decoding algorithms). Indeed, let $\bfY = \bfC + \bfE$ be
  a corrupted codeword of an $u$--interleaved Gabidulin code, with
  $\bfE\in\Fqm^{u\times n}$ being an error matrix of $\Fq$-rank equal
  to $t$. Then, the rows of $\bfE$ as seen as vectors of $\Fqm^{n}$
  share a common \emph{row} support, namely the \emph{row} support of
  $\bfE$, of dimension $t$. Hence, they also share a common right
  annihilator, of $q$--degree at most $t$. The algorithm from \cite[\S
  4]{LO06} where the errors shared a common \emph{column} support can
  be adapted straightforwardly in this setting, and allows to decode
  almost all error matrix $\bfE$ of rank-weight
  $t\le \lfloor\frac{u}{u+1}(n-k)\rfloor$. This can be used in order
  to attack the original Faure--Loidreau cryptosystem in the same
  fashion as \cite[\S 3]{RPW20}.
\end{remark}
\fi

 \else
In
Appendix~\ref{sec:appendix_decoding}, we give further details on how
to implement in practice such a decoder.
We believe that this algorithm might be folklore, but
we weren't able to find it in the literature.
\fi

 \section{Decoding Supercodes of Gabidulin Codes}\label{sec:supercode}
A common feature of the cryptanalyses to follow can be understood
as the decoding of a supercode of a Gabidulin code.
Consider a code (represented as a subspace of $\qpoly{<m}$)
\[
  \CC \eqdef \qpoly{<k} \oplus \TT,
\]
where $\TT \subseteq \qpoly{<m}$, the code $\CC$ benefits from a
decoding algorithm in a similar manner to that of \cite{L06a}.
Indeed, given a received word
\if\longversion1
\[
  \else
  \(
  \fi
  Y = C + E
  \if\longversion1
\]
\else
\)
\fi
where $C \in \CC$ and $E \in \qpoly{<m}$ with $\rk{E} = t$, one can
look for the left annihilator of $E$. Let $\Lambda \in \qpoly{\leq
  t}$ be the left annihilator of $E$. We have to solve
\[
  \Lambda \circ Y \equiv \Lambda \circ C \mod (X^{q^m}-X),
\]
where the unknowns are $\Lambda, C$. Then, similarly to the decoding
of Gabidulin codes, one may linearise the system. For this sake, recall
that $C = C_0 + T$ for $C_0 \in \qpoly{<k}$ and $T \in \TT$.
Therefore, we are looking for the solutions of a system
\begin{equation}\label{item:supercode_system}
  \Lambda \circ Y \equiv N \mod (X^{q^m}-X)
\end{equation}
where $N \in (\qpoly{\leq t} \circ \qpoly{<k}) +\qpoly{\leq t} \circ \TT =
\qpoly{<k+t} + \qpoly{\leq t} \circ \TT$.

\begin{lemma}\label{lem:super_code_decoding_radius}
  Under the assumption that
  $(\qpoly{<k+t} + \qpoly{\leq t}\circ \TT) \cap (\qpoly{\leq t} \circ E)
  = \{0\}$, any nonzero solution $(\Lambda, N)$ of the
  system~(\ref{item:supercode_system}) satisfies
  $\Lambda \circ E = 0$.
\end{lemma}

\begin{proof}
  Let $(\Lambda, N)$ be such a nonzero solution.
  Then,
  \[
    \Lambda \circ Y - \Lambda \circ C \equiv \Lambda \circ E \mod
    (X^{q^m} - X).
  \]
  Since $\Lambda \circ Y \equiv N \mod (X^{q^m}-X)$, the left--hand
  side is contained in $(\qpoly{<k+t} + \qpoly{\leq t}\circ \TT)$
  while the right--hand one is contained in $\qpoly{\leq t} \circ E$.
  Therefore, by assumption, both sides are zero. This yields the
  result.
\end{proof}

Under the hypotheses of Lemma~\ref{lem:super_code_decoding_radius},
decoding can be performed as follows.
\begin{enumerate}
\item Solve System~(\ref{item:supercode_system}).
\item Take a nonzero solution $(\Lambda, N)$ of the system.  Compute
  the right kernel of $\Lambda$. This kernel contains the image of $E$
  and hence the support of the error.
\item Knowing the support of $E$, one can recover it by solving a linear
  system. See for instance \cite[\S 3]{GRS13} or \cite[\S
  1.4]{AGHT17}.
\end{enumerate}

\begin{remark}
  Note that for the decoding of Gabidulin codes,
  once a solution $(\Lambda, N)$ is computed, one can recover
  $C$ by left Euclidean division of $N$ by $\Lambda$. In the present
  situation, this calculation is no longer efficient. Indeed,
  the proof of Lemma~\ref{lem:super_code_decoding_radius} permits
  only to assert that $N \equiv \Lambda \circ C \mod (X^{q^m} - X)$.
  In the Gabidulin case, the fact that $\qdeg C < k$ permits
  to assert that $\qdeg \Lambda \circ C < m$ and hence that
  $N = \Lambda \circ C$. This is no longer true in our setting
  since, there is {\em a priori} no upper bound on the
  $q$--degree of $C$. For this reason, we need to use the knowledge
  of the support of the error to decode.
\end{remark}

For the decoding to succeed, the condition
$(\qpoly{<k+t} + \qpoly{\leq t}\circ \TT) \cap (\qpoly{\leq t} \circ
E) = \{0\}$ needs to be satisfied. In the case of Gabidulin codes
(\ie{} $\TT = \{0\}$), this is
guaranteed by a minimum distance argument entailing that
$\qpoly{< k+t} \cap \qpoly{\leq t} \circ E$ is zero as soon as
$t \leq \frac{n-k}{2}$. In our situation, estimating the minimum
distance of $\qpoly{< k+t} + \qpoly{\leq t} \circ \TT$ is
difficult. However, one can expect that in the typical case, the
intersection
$(\qpoly{< k+t} + \qpoly{\leq t} \circ \TT) \cap (\qpoly{\leq t} \circ
E)$ is $0$ when the sums of the dimensions of the codes is
less than that of the ambient space.  Therefore, one can reasonably
expect to correct almost any error of rank $t$ as soon as
\begin{equation}
  \label{eq:condition_supercode_decoding}
  k + 2t + \dim (\qpoly{\leq t} \circ \TT) \leq n.
\end{equation}
In the case $\TT = \{0\}$, we find back the decoding radius of
Gabidulin codes.

\if\longversion1
\begin{remark}
Note that the previous approach applies {\em mutatis mutandis}
to the decoding of supercodes of Reed--Solomon codes.
\end{remark}
\fi

\paragraph{The right--hand side version.} In the spirit of
Section~\ref{subsec:right_decoding}, a similar approach using
right--hand side decoding shows that decoding is also possible when
\begin{equation}
  \label{eq:condition_supercode_right_decoding}
  k + 2t + \dim (\TT \circ \qpoly{\leq t}) \leq n.
\end{equation}

\if\longversion1
We will use this decoding algorithm to attack \LIGA{}
and \RAMESSES{}.  For \LIGA{}, the code $\TT$ is a random code of
low dimension, while for \RAMESSES{}, $\TT$ is a Gabidulin code.
\fi

 \section{Applications to Cryptanalysis}
\subsection{\RAMESSES{}}
\if\longversion1
The decoder introduced in Section~\ref{sec:supercode} is the key of
our cryptanalysis of \RAMESSES{}.
\fi
Using the notation of
Section~\ref{subsec:ramesses}, suppose we have a ciphertext as
in (\ref{eq:ramesses_ciphertext_rewritten}):
\[
  Y  = C+E \quad {\rm with } \quad C = C_1 + \Ksec \circ T ,
\]
where $C_1 \in \qpoly{<k+\ell}$, $T \in \qpoly{\ell}$ and
$E \in \qpoly{<m}$ of rank $t$. Recall that the plaintext is the row
space of $E$ (equivalently, the image of $\adj{E}$).  We perform the
right--hand side version of the decoding algorithm of
Section~\ref{sec:supercode}. Here the code $\TT$ is
$\Ksec \circ \qpoly{\leq \ell}$ and the supercode $\CC$ is
$\qpoly{<k+\ell} + \TT$.
\if\longversion1
Therefore, cryptanalysis will consist in decoding
a supercode of a Gabidulin code of dimension $k+\ell$.
\fi
We compute the solutions $(\Lambda, N)$ of the
system
\[
  Y \circ \Lambda \equiv N \mod (X^{q^m} - X),
\]
where
\[ N \in \CC \circ \qpoly{\leq t} = \qpoly{<k+\ell+t} + \Ksec \circ
  \qpoly{\leq \ell+t}.
\]
According to Lemma~\ref{lem:super_code_decoding_radius} and
(\ref{eq:condition_supercode_right_decoding}), the algorithm will very
likely return pairs of the form $(\Lambda, C \circ \Lambda)$ with
$E \circ \Lambda = 0$ as soon as
\begin{equation}\label{eq:parameters_to_break_ramesses}
  k + \ell + 2t + \dim (\TT \circ \qpoly{\leq t}) =
  k +\ell + 2t + \dim (\Ksec \circ
  \qpoly{\leq t + \ell}) = k + 3t + 2\ell + 1 \leq n.
\end{equation}
Once such a $\Lambda$ is obtained, one recovers $E$
and the image of $\adj{E}$ yields the plaintext.

A comparison of (\ref{eq:parameters_to_break_ramesses}) with
the proposed parameters for \RAMESSES{} in \cite[Section~4]{LLP20} is
given in Table~\ref{tab:table_Ramesses}. As observed,
inequality~(\ref{eq:parameters_to_break_ramesses}) is satisfied for
any proposed parameter set.

\begin{table}[!h]
  \begin{center}
  \begin{tabular}{|c|c|c|c|c|c||c|}
  \hline
    \ \ $m\ (=n)$ \ \ & \ \ $k$ \ \ & \ \ $w$ \ \ & \ \ $\ell$ \ \ &
    \ \ $t$ \ \ & \ \ Security (bits)\ \ & \ \ $k+3t+2\ell + 1$ \ \ \\
  \hline \hline
  64 & 32 & 19 & 3 & 5 & 141 & 54 \\
  80 & 40 & 23 & 3 & 7 & 202 & 68 \\
  96 & 48 & 27 & 3 & 9 & 265 & 82 \\
  \hline
  164 & 116 & 27 & 3 & 9 & 256 & 150 \\
  \hline
\end{tabular}
\medskip
\caption{\label{tab:table_Ramesses} This table compares the values of the
  formula~(\ref{eq:parameters_to_break_ramesses}) with the parameters proposed
  for \RAMESSES{}. The first three rows are parameters for \RAMESSES{} as a KEM
  and the last one are parameters for \RAMESSES{} as a PKE. Note that for any
  proposed parameter set, we have $m=n$.}
\end{center}
\end{table}

 \subsection{A Message Recovery Attack Against \LIGA{} Cryptosystem}

In this section, we show that it is possible to recover the plaintext from a ciphertext. Notice that \LIGA{} cryptosystem has been proven IND-CCA2 in \cite{RPW20}, under some computational assumption, namely the Restricted Gabidulin Code Decision Problem (\cite{RPW20}, Problem 4). We are not disproving this claim here, however our attack can be precisely used as a distinguisher, and hence this problem is not as hard as supposed.

Recall that $\bfG$ is a generator matrix of a public Gabidulin code
$\Gab{k}{\bfg}$, the public key is a noisy vector
$\kpub = \bfx \cdot \bfG + \bfz$ and the encryption of a message
$\bfm$ is $\bfc = \bfm\cdot \bfG + \Tr(\alpha\kpub) + \bfe$ for some
uniformly random element $\alpha\in\Fqmu$ and a uniformly random error
$\bfe$ of small rank weight
$t_{pub} = \lfloor \frac{n-k-w}{2} \rfloor$ both chosen by Alice.
See Section~\ref{sec:FL_and_LIGA} for further details.

The attack works in two parts. First, we introduce a supercode of the public Gabidulin code, in which we are able to decode the ciphertext and get rid of the small error. Then, we recover the plaintext.

\subsubsection{Step 1: Get rid of the small error.}
Let $\zeta \eqdef \rk[q^{m}](\bfz)$, so that
$\bfz = \sum_{i=1}^{\zeta}\mu_{i}\bfz_{i}$ where the $\mu_{i}$'s $\in\Fqmu$ and
the $\bfz_{i}$'s $\in\Fqm^{n}$ are both linearly independent over $\Fqm$. The
ciphertext can now be written as
\begin{equation}
  \bfc = \bfm\cdot\bfG + \sum_{i=1}^{\zeta}\Tr(\alpha\mu_{i})\bfz_{i}
  + \bfe
\end{equation}
Let

\begin{equation}\label{eq:code_C}
  \CC \eqdef \Gab{k}{\bfg} + \span_{\Fqm} \{\bfz_1, \dots, \bfz_\zeta\} \subseteq
  \Fqm^n.
\end{equation}

The ciphertext can be seen as a
codeword of $\CC$ corrupted by a small rank weight error
$\bfe$. Moreover, $\CC$ can be computed from public data as
suggested by the following statement.

\begin{theorem}\label{thm:prob_C_eq_Cpub}
  Let $\CC$ be the code defined in (\ref{eq:code_C}) and $\CC_{pub}$
  be the code generated by $\Gab{k}{\bfg}$ and $\Tr(\gamma_{i}\kpub)$
  for $i\in \{1,\dots,\zeta\}$, where the $\gamma_{i}$'s denote $\zeta$
  elements of $\Fqmu$ linearly independent over $\Fqm$. Then, for a
  uniformly random choice of $(\gamma_1, \dots, \gamma_{\zeta})$,
  \[
    \Prob(\CC = \CC_{pub}) = 1 - e^{O\left(\frac{1}{q^m}\right)}.
  \]
\end{theorem}

The proof of Theorem~\ref{thm:prob_C_eq_Cpub} rests on the following
technical lemma.

\begin{lemma}\label{lemma:denombrement}
  Let $F$ be a linear subspace of dimension $m$ in a linear space $E$
  of dimension $n$ over a finite field $\Fq$. Then,
  $\#\{G \mid F \oplus G = E\} = q^{m(n-m)}$.
\end{lemma}

\begin{proof}
  Let $\Stab(F)$ denote the stabiliser of $F$ under the action of
  $\textrm{GL}(E)$. It is isomorphic to the group of the matrices of the form $\begin{pmatrix}\bfA & \bfB \\ \bigzero & \bfC\end{pmatrix}$ with $\bfA \in \GL{m}{\Fq}$, $\bfC\in\GL{n-m}{\Fq}$ and $\bfB\in\Mspace{m, n-m}{\Fq}$, \ie{}
  $$\Stab(F) \cong (\GL{m}{\Fq}\times\GL{n-m}{\Fq})\ltimes
  \Mspace{m,n-m}{\Fq}.$$ 

\noindent            This group acts transitively on the complement spaces of $F$. Indeed, let $G$ and $G'$ be such that $F\oplus G=F\oplus G'=E$. Let $(f_{1}, \dots, f_{m})$ be a basis of $F$ and $(g_{1},\dots,g_{n-m})$ (respectively $(g_{1}',\dots,g_{n-m}')$) be a basis of $G$ (resp. $G'$). Then the linear map that stabilises $F$ and maps $g_{i}$ onto $g_{i}'$ is an element of $\Stab(F)$ that maps $G$ onto $G'$. The stabiliser of a complement $G$ under this action is simply $\GL{m}{\Fq}\times\GL{n-m}{\Fq}.$ Hence, $$\#\{G \mid F \oplus G = E\} = \dfrac{(\# \GL{m}{\Fq}) \times (\#\GL{n-m}{\Fq}) \times q^{m(n-m)}}{(\# \GL{m}{\Fq}) \times (\#\GL{n-m}{\Fq})} = q^{m(n-m)}.$$
\end{proof}

{
  \renewenvironment{proof}{\noindent {\itshape Proof of Theorem~\ref{thm:prob_C_eq_Cpub}.}}{\qed \bigskip}

  \begin{proof}
    We wish to estimate the probability that $\CC = \Cpub$.
    Note first that inclusion $\supseteq$ is always satisfied.
  \if\longversion1
 Indeed, let
  $\bfc\in\CC_{pub}$. Then, there exist $\bfm\in\Fqm^{k}$ and
  $\lambda_{1},\dots,\lambda_{\zeta}\in\Fqm$ such that
          \begin{equation*}
            \begin{array}{ccc}
              \bfc & = \bfm\bfG + \displaystyle\sum_{i=1}^{\zeta}\lambda_{i}\Tr(\gamma_{i}\kpub) & \\
                   & = \left(\bfm + \displaystyle\sum_{i=1}^{\zeta}\lambda_{i}\Tr(\gamma_{i}\bfx)\right)\bfG & + \displaystyle\sum_{i=1}^{\zeta}\sum_{j=1}^{\zeta}\lambda_{j}\Tr(\gamma_{j}\mu_{i})\bfz_{i}
            \end{array}
          \end{equation*}
          and $\bfc\in\CC$.
  \else
  This can be checked by an elementary calculation.
  \fi

  Therefore, the following equality of events holds:

  $$(\CC = \Cpub) = (\CC \subseteq \Cpub),$$

  \noindent and we are reduced to study the probability that
  $\CC \subseteq \Cpub$.

 Let $\bfc\in\CC$. There
  exists $\bfm\in\Fqm^{k}$ and
  $\lambda_{1},\dots,\lambda_{\zeta}\in\Fqm$ such that
  $$\bfc = \bfm\bfG +
  \displaystyle\sum_{i=1}^{\zeta}\lambda_{i}\bfz_{i}.$$
  If we can find
  $\mathbf{\alpha} \eqdef
  (\alpha_{1},\dots,\alpha_{\zeta})\in\Fqm^{\zeta}$ such that
  $\bfc - \displaystyle\sum_{i=1}^{\zeta}
  \alpha_{i}\Tr(\gamma_{i}\kpub) \in \Gab{k}{\bfg}$, then we are done.
          \begin{equation*}
            \begin{array}{cc}
              \bfc - \displaystyle\sum_{i=1}^{\zeta} \alpha_{i}\Tr(\gamma_{i}\kpub) = & \\
              \left(\bfm - \displaystyle\sum_{i=1}^{\zeta}\alpha_{i}\Tr(\gamma_{i}\bfx)\right)\bfG + \displaystyle\sum_{i=1}^{\zeta}\left( \lambda_{i} - \sum_{j=1}^{\zeta}\alpha_{j}\Tr(\gamma_{j}\mu_{i})\right)\bfz_{i}.
            \end{array}
          \end{equation*}
          It suffices to choose $\alpha$ such that $\lambda_{i} - \sum_{j=1}^{\zeta}\alpha_{j}\Tr(\gamma_{j}\mu_{i}) = 0$ for $i\in \{1,\dots,\zeta\}$, \ie

          $$(\lambda_{1},\dots,\lambda_{\zeta}) = (\alpha_{1}, \dots, \alpha_{\zeta})\begin{pmatrix} \Tr(\gamma_{1}\mu_{1}) & \cdots & \Tr(\gamma_{1}\mu_{\zeta}) \\ \vdots & \ddots & \vdots \\ \Tr(\gamma_{\zeta}\mu_{1}) & \cdots & \Tr(\gamma_{\zeta}\mu_{\zeta})\end{pmatrix}.$$
          Let $\bfM$ denote this last matrix. The previous remark implies $\Prob(\CC \subseteq \CC_{pub}) \ge \Prob(\bfM \text{ is non singular })$, therefore it suffices to prove that $\bfM$ is non singular with overwhelming probability over the choice of $\gamma_{1},\dots,\gamma_{\zeta}$.

          Let
          \[ \Gamma \eqdef \span(\gamma_{1}, \dots, \gamma_{\zeta}) \quad {\rm and} \quad \MM \eqdef \span(\mu_{1},\dots,\mu_{\zeta}).\] Then, $\bfM$ is singular if and only if $\Gamma \cap \MM^{\perp} \neq \{0\}$. 
          Since $\Gamma$ and $\MM$ have the same dimension $\zeta$ over $\Fqm$, $\Gamma \cap \MM^{\perp} = \{0\}$ if and only if $\Gamma \oplus \MM^{\perp} = \Fqmu$. Therefore,

          $$\Prob(\bfM \text{ is non singular }) = \dfrac{\#\{\Gamma \mid \MM^{\perp} \oplus \Gamma = \Fqmu\}}{\#\{\Gamma \mid \dim_{\Fqm}(\Gamma)=\zeta\}}\cdot$$

          Recall the Gaussian binomial coefficient
          $\gbinom{u}{\zeta}_{q^{m}}$ denotes the number of
          $\Fqm$--linear subspaces of dimension $\zeta$ in an $\Fqm$--vector
          space of dimension $u$. Applying
          Lemma~\ref{lemma:denombrement}, we have
          $$\Prob(\bfM \text{ is non singular}) = \dfrac{q^{m\zeta(u-\zeta)}}{\gbinom{u}{\zeta}_{q^{m}}}  \ge \left(1 - \dfrac{1}{q^{m}}\right)^{\dfrac{q^{m}}{q^{m}-1}},$$
          where the inequality on the right--hand side can be found
          for instance in \cite[Appendix~A]{CC19}. This yields Theorem~\ref{thm:prob_C_eq_Cpub}.
\end{proof}
}

Set
\if\longversion1
\[
  \else
  \(
  \fi
  \TT \eqdef \bigoplus_{i=1}^{\zeta}\span_{\Fqm}\left\{ \Tr(\gamma_{i}\kpub)\right\}.
  \if\longversion1
\]
\else
\)
\fi
By interpolation, it can be regarded as a subspace of $\qpoly{<m}$, and $\CC_{pub} = \qpoly{<k} \oplus \TT$. In order to remove the error $\bfe$ we just need to decode in this public supercode. Notice that $$\dim(\qpoly{\le t}\circ\TT) \le \zeta(t+1).$$ Therefore, using the algorithm of Section \ref{sec:supercode}, one can expect to decode in $\CC_{pub}$ whenever
\begin{equation}\label{eq:parameters_to_break_liga}
k+2t+\zeta(t+1) \le n.
\end{equation}

Table~\ref{tab:table_liga} compares \eqref{eq:parameters_to_break_liga} with the proposed parameters for \LIGA{} in \cite[Section~7]{RPW20}. As observed, Inequality~(\ref{eq:parameters_to_break_liga}) is satisfied for
any proposed parameter set. Moreover, if one tries to increase $\zeta$ in order to avoid this attack, one also needs to increase $w$ to resist the key recovery attack from \cite{GOT18}, which decreases $t \eqdef \lfloor \frac{n-k-w}{2} \rfloor$ that must be greater than 1.

\begin{table}[!h]
  \begin{center}
  \begin{tabular}{|c||c|c|c|c|c||c|}
  \hline
    \ \ \textbf{Name}\ \ & \ \ $n$ \ \ & \ \ $k$ \ \ & \ \ $t$ \ \ & \ \ $\zeta$ \ \ & \ \ Security (bits)\ \ & \ \ $k+2t+\zeta(t+1)$ \ \ \\
  \hline \hline
  \LIGA-128 & 92  & 53 & 6  & 2 & 128 & 79  \\
  \LIGA-192 & 120 & 69 & 8  & 2 & 192 & 103 \\
  \LIGA-256 & 148 & 85 & 10 & 2 & 256 & 127 \\
  \hline
\end{tabular}
\medskip
\caption{\label{tab:table_liga} This table compares the values of the
  formula~(\ref{eq:parameters_to_break_liga}) with the parameters proposed for
  \LIGA{}.}
\end{center}
\end{table}

Step 1 is summed up in Proposition \ref{prop:sum_up_step_1_liga}.

\begin{proposition}\label{prop:sum_up_step_1_liga}
  If $\bfc = \bfm\cdot\bfG + \Tr(\alpha\kpub) + \bfe$ is the encryption of a plaintext $\bfm$, then we can recover the support of the error $\bfe$ and the corrupted codeword $\bfm\cdot\bfG + \Tr(\alpha\kpub)$ in polynomial time using only the knowledge of the public key.
\end{proposition}

\subsubsection{Step 2: Remove the $\bfz$ dependency.}
From now on, since we got rid of the small error term $\bfe$,
we can do as if the ciphertext was
\begin{equation}
  \begin{aligned}
    \bfc' & \eqdef \bfm\cdot\bfG + \Tr(\alpha\kpub) \\
         & = (\bfm + \Tr(\alpha\bfx))\cdot\bfG + \Tr(\alpha\bfz).
  \end{aligned}
\end{equation}
This is a codeword of a Gabidulin code $\GG \eqdef \Gab{k}{\bfg}$, corrupted by an error of rank $w > \lfloor \frac{n-k}{2}\rfloor$. Hence, we cannot decode in $\GG$ to recover the plaintext. However, thanks to the knowledge of the public key, one can easily recover the affine space
\[
  A \eqdef \{ \beta\in\Fqmu \mid \bfc' - \Tr(\beta\kpub) \in \GG\}
\]
using linear algebra.

\begin{lemma}
  Let $\beta\in\Fqmu$. Then $\bfc' - \Tr(\beta\kpub)\in \GG$ if and only if\\
  $\Tr((\alpha-\beta)\bfz)~=~0$.
\end{lemma}

\begin{proof}
  Note that for any $\lambda\in\Fqmu$, $$\rk(\Tr(\lambda \bfz)) \le \rk(\bfz) = w < n-k.$$ Indeed, let $\BB$ be a basis of the extension field $\Fqmu/\Fq$. Then, if $\lambda \neq 0$, the extension of $\lambda\bfz$ in $\BB$ is the extension of $\bfz$ in the basis $\lambda\BB$. Therefore, $$\RSupp(\lambda\bfz) = \RSupp(\bfz)$$ and the trace cannot increase the rank.

  \noindent Let $\beta\in\Fqmu$. Then $$\bfc' - \Tr(\beta\kpub) = (\bfm + \Tr((\alpha-\beta)\bfx))\bfG + \Tr((\alpha-\beta)\bfz).$$ Therefore, $\beta\in A$ if and only if $\Tr((\alpha-\beta)\bfz)\in\GG$. Since it has rank weight less than the minimum distance of $\GG$, it follows that $\Tr((\alpha-\beta)\bfz) = 0$.
\end{proof}

\begin{lemma}
  Let $\EE\eqdef \displaystyle \bigcap_{i=1}^{\zeta}\langle \mu_{i} \rangle^{\perp}$. Then $A$ is the affine space $\alpha + \EE$.
\end{lemma}

\begin{proof}
  $\beta\in A$ if and only if $$\Tr((\alpha - \beta)\bfz) = \sum_{i=1}^{\zeta}\Tr((\alpha-\beta)\mu_{i})\bfz_{i} = 0.$$ By the linear independence of the $\bfz_{i}$'s, it follows that $\Tr((\alpha-\beta)\mu_{i}) = 0$ for all $i$, \ie
  \if\longversion1
  \[
  \else
  \( 
  \fi
  A = \alpha + \bigcap_{i=1}^{\zeta}\langle \mu_{i} \rangle^{\perp}.
  \if\longversion1
  \]
  \else
  \)
  \fi
\end{proof}

We are now able to remove the $\bfz$ dependency in the ciphertext. Indeed, let $\FF \eqdef \{\Tr(\gamma\bfx) \mid \gamma\in\EE\}.$ The knowledge of $A$ gives finally access to the affine space $\bfm + \FF$.

\subsubsection{Step 3: Recover the plaintext.}
Denote by $f \eqdef \dim_{\Fqm} \FF$. Since $\FF$ is the image of $\EE$
by a surjective map, we have $f\leq \dim \EE =  u- \zeta \leq u-1$.
Let $\bfs$ be some random element of $\bfm+\FF$. Notice that from a description of the affine space $\bfm + \FF$ it is possible to recover a basis $(\bfe_{1}, \dots, \bfe_{f})$ of $\FF$. Then, $\bfs$ can be decomposed as
$$\bfs \eqdef \bfm + \displaystyle\sum_{i=1}^{f} \lambda_{i}\bfe_{i}$$ for some unknown coefficients $\lambda_{i}\in\Fqm$. Furthermore, recall that the last $u$ positions of $\bfm$ are $0$. Then, $\bfm$ is a solution of the following linear system of $k+f$ unknowns and $u+k$ equations:
\begin{equation}\label{eq:systeme_message_liga}
		\left\{
		\begin{array}{c}
			\bfm + \displaystyle\sum_{i=1}^{f} \lambda_{i}\bfe_{i} = \bfs \\
			\bfm_{k-u+1} = \dots = \bfm_{k} = 0
		\end{array}\right.
\end{equation}
Finally, the following lemma shows that $\bfm$ can be recovered from \emph{any} solution of (\ref{eq:systeme_message_liga}).

\begin{lemma}
  Let $(\bfm', \mathbf{\lambda'})$ be another solution of (\ref{eq:systeme_message_liga}). Then $\bfm' = \bfm$.
\end{lemma}

\begin{proof}
  Since
  $\bfm - \bfm' = \sum_{i=1}^{f}(\lambda_{i}'-\lambda_{i})\bfe_{i}
  \in \FF$, it is of the form $\Tr(\gamma\bfx)$ for some
  $\gamma \in \EE$. Moreover, its last $u$ positions are $0$. Recall
  that $(\bfx_{k-u+1},\dots,\bfx_{k})$ is a basis of
  $\Fqmu/\Fqm$. Then, the last $u$ positions of $\Tr(\gamma\bfx)$ are
  the coefficients of $\gamma$ in the dual basis
  $\{\bfx^{\ast}_{k-u+1},\dots,\bfx_{k}^{\ast}\}$. Hence, $\gamma = 0$
  and $\bfm = \bfm'$.
\end{proof}

\subsubsection{Summary of the attack}
\begin{itemize}
  \item Decode in a public supercode of a Gabidulin code to get rid of the small error $\bfe$ and recover $\bfc' = \bfm\bfG + \Tr(\alpha\kpub)$.
  \item Using linear algebra, deduce the affine space $$A = \{ \beta\in\Fqmu \mid \bfc' - \Tr(\beta\kpub)~\in~\GG\}.$$
  \item Recover the affine space $\bfm + \FF$ where $\FF = \{ \Tr(\gamma\bfx) \mid \alpha + \gamma \in A \}$.
  \item Deduce a basis of $\FF$.
  \item Solve linear system (\ref{eq:systeme_message_liga}) to recover
    the plaintext $\bfm$.
\end{itemize}

\subsubsection{Implementation.}
Tests have been done using SageMath v9.2 \cite{sage9.2} on an
Intel$^{\text{\textregistered}}$ Core™ i5-10310U CPU. We are able to
recover the plaintext on the three \LIGA{} proposals. The average
running times are listed in Table \ref{table:benchmark_liga}. Our
implementation is available on Github
\url{https://github.com/mbombar/Attack_on_LIGA}.

\begin{table}[ht]
  \centering
  \begin{tabular}{|c||c|c|c|}
  \hline
  \textbf{Name} & \begin{tabular}{@{}c@{}}
          \textbf{Parameters} \\ $\mathbf{(q, n, m, k, w, u, \zeta)}$
        \end{tabular}  & \begin{tabular}{@{}c@{}}
          \textbf{Claimed security} \\ \textbf{level}
        \end{tabular} & \begin{tabular}{@{}c@{}}
          \textbf{Average running time}
        \end{tabular} \\
  \hline \hline
  \LIGA-128 & $(2, 92, 92, 53, 27, 5, 2)$ & 128 bits & 8 minutes \\
  \hline
  \LIGA-192 & $(2, 120, 120, 69, 35, 5, 2)$ & 192 bits & 27 minutes \\
  \hline
  \LIGA-256 & $(2, 148, 148, 85, 43, 5, 2)$ & 256 bits & 92 minutes \\
  \hline
  \end{tabular}
  \medskip
\caption{\label{table:benchmark_liga} Average running times for the attack on \LIGA.}
\end{table}

\subsubsection{Acknowledgements.} The second author is partially funded by the ANR project 17-CE39-0007
{\em CBCrypt}.

\appendix
\section{Further Details About RAMESSES'
  Specifications}\label{app:ramesses}
As explained in Section~\ref{subsec:ramesses}, our presentation of
\RAMESSES{} may seem to differ from the original proposal
\cite{LLP20}. Indeed in Section~\ref{subsec:ramesses}, we present the
scheme using only $q$--polynomials, while the original publication
prefers using matrices and vectors.
The purpose of the present appendix is to prove that our way to
present \RAMESSES{} is equivalent to that of \cite{LLP20}.

\medskip

\noindent {\bf Caution.} In the present article we use $q$ to denote
the cardinality of the ground field $\Fq$. In \cite{LLP20} the ground
field is always supposed to be $\F_2$ and $q$ refers to some power of
$2$, {\ie} $q = 2^n$ for some positive $n$.  Moreover, the exponent of
$q$ is denoted $n$ while it is denoted $m$ in the present article.
This might be confusing while reading both papers in parallel.

The other notations $w, k, \ell, t$ are the same in the two articles.
Finally, in \cite{LLP20} a public $\Fq$--basis $\bfg = (g_1, \dots, g_m)$
of $\Fqm$ is fixed once for all. Our presentation does not require
such a setting.

\subsection{Key Generation.}
Recall that \cite{LLP20} fixes a vector $\bfg \in \Fqm^m$ of rank $m$
({\ie} an $\Fq$--basis of $\Fqm$). This data together
with a parity--check matrix $\bfH$ of the Gabidulin code
$\Gab{k}{\bfg}$ are public.

\paragraph{Original presentation}
The key generation consists in picking a uniformly random
$\kpriv \in \Fqm^m$ of weight $w$ and the public key is its syndrome
$\kpub \eqdef \bfH \kpriv^\top$ with respect to the public Gabidulin
code.

\paragraph{Our presentation}
Since the code $\Gab{k}{\bfg}$ is public,
any of its elements may be associated to an element of
$\qpoly{<k}$. The transition from codewords to $q$--polynomials is
nothing but interpolation.
Similarly, the choice of a vector $\kpriv \in \Fqm^m$ of rank $w$ is
(again by interpolation) equivalent to that of a $q$--polynomial $K$
of rank $w$.  Finally, publishing its syndrome $\bfH \kpriv^{\top}$ is
equivalent to publish the coset $\kpriv + \Gab{k}{\bfg}$, which in our
setting is nothing but publishing the affine space
$\Ksec + \qpoly{<k}$.

\subsection{Encryption}

\paragraph{Original presentation}
The plain text is encoded into a matrix $\bfP \in \Fq^{m \times m}$
in row echelon form and of rank $t$.
\begin{itemize}
\item Compute $\bfy \in \Fqm^m$ such that $\bfH \bfy^\top = \kpub$
\item Pick a uniformly random $\bfT \in \Fq^{m \times m}$ of
  $\bfg$--degree $\ell$ \ie{} representing a $q$--polynomial of
  $q$--degree $\ell$ in the basis $\bfg$;
\item Pick a uniformly random \(\bfS \in \GL{m}{\Fq}\).
\end{itemize}
The ciphertext is
\[
  \bfu^{\top} \eqdef \bfH (\bfy \bfT + \bfg \bfS \bfP)^{\top}.
\]

\paragraph{Our presentation}
The vector $\bfu$ is a syndrome of any word of the form:
\[
  \bfy \bfT + \bfg \bfS \bfP + \bfc,
\]
where $\bfc$ ranges over $\Gab{k}{\bfg}$.
From a $q$--polynomial point of view, such a word corresponds to:
\[
  (\Ksec + C_0) \circ T + G \circ S \circ P + C,
\]
where
\begin{itemize}
\item $C, C_0$ are arbitrary elements of $\qpoly{<k}$;
\item $T \in \qpoly{\ell}$ is the interpolating polynomial of $\bfT$;
\item and $G, S, P$ are the respective interpolating polynomials of $\bfg, \bfS, \bfP$.
\end{itemize}
Note that, since $\bfg$ has rank $m$ and $\bfS$ is supposed to be
nonsingular, then their interpolating polynomials are invertible in
$\qpoly{}/( X^{q^m} - X )$. Hence, setting $E \eqdef G \circ S \circ P$, we
get a $q$--polynomial whose matrix representation in basis $\bfg$ has
the same row space as the matrix representation of $P$.  Thus, we get
the ciphertext description in (\ref{eq:ramesses_ciphertext}).

\subsection{Decryption}

\paragraph{Original presentation}
Start by computing $\bfx \in \Fqm^n$ such that
$\bfH \bfx^\top = \bfu^\top$.  Next, knowing $\kpriv$, one can compute
an annihilator polynomial $V_{\kpriv} \in \qpoly{w}$ of the support of
$\kpriv$.  Then, compute
$\bfz \eqdef V_{\kpriv}(\bfx) = (V_{\kpriv}(x_1), \dots,
V_{\kpriv}(x_m))$ and decode $\bfz$ as a corrupted codeword of
$\Gab{k+\ell + w}{\bfg}$.  If succeeds, it returns an error vector
$\bfa$. If its rank equals $t$, then the row echelon form of
$\ext{\bfa}{\bfg}$ yields $\bfP$.

\paragraph{Our presentation}
Similarly, the approach is based on applying $V_{\kpriv}$ and
performing Gabidulin codes decoding.  Indeed, starting from ciphertext
(\ref{eq:ramesses_ciphertext}), we apply $V = V_{\kpriv}$ and get
\[
  V \circ Y \equiv V \circ C_1 + V \circ E
\]
and a decoding procedure returns $V \circ E$. If this $q$--polynomial
has rank $t$, then the row echelon form of its matrix representation
yields the plaintext $\bfP$.

\if\longversion0
\section{Detailed Presentation of the
  Right--Hand Side Version of the Decoding of Gabidulin Codes}\label{sec:appendix_decoding}
In this appendix, we provide a detailed and self--contained version of
the alternative decoder for Gabidulin codes presented in
Section~\ref{subsec:right_decoding}.

Starting from the decoding problem $Y = C+E$, the decoding problem can be
thought as finding the $q$--polynomial $C$, given $Y$. Using the
analogy with Reed--Solomon codes, Loidreau introduced in \cite{L06a} a
Welch--Berlekamp like algorithm to decode Gabidulin codes that
consists in finding the unique $q$--polynomial $V$ of $q$--degree less
than or equal to $t$ such that $V$ vanishes on the column support of
$\bfe$, which is equivalent to $V\circ E=0$, \ie $V$ is a left
annihilator of the error. Using a linearisation technique, this leads
to the resolution of a linear system that can be efficiently solved
provided that $t$ is less than half the minimum distance. It then
suffices to compute a left Euclidean division to recover $C$ and
therefore the codeword $\bfc$.

The core of the algorithm to follow consists in searching a
right--hand side annihilator of $E$ instead of a left--hand side one.  Due
to the non commutativity of the ring $\qpoly{}$, working on the
right--hand side is not directly equivalent to working on the left--hand side.

We begin to state the existence of a right--hand side annihilator.

\begin{proposition} Let $E$ be a $q$--polynomial of rank $t$. Then there
exists a unique monic $q$--polynomial $V$ with $\qdeg(V)\le t$ such that $E\circ V=0$ modulo
$(X^{q^m}-X)$.
\end{proposition}

\begin{proof} Let $Q \eqdef \sum_{i=0}^{t}a_{i}X^{q^{i}}$ be the unique monic $q$--polynomial of $q$--degree less than or equal to
  $t$ that vanishes exactly on $\img(\adj{E})$, \ie $\img(\adj{E}) = \ker Q$. Such a
  $q$--polynomial is guaranteed to exist (see for instance \cite{B15a} or \cite{O33}.) It follows that $\ker(E)=\img(\adj{Q})$.
  Moreover, $$\adj{Q} = \sum_{i=0}^{t}a_{i}^{q^{m-i}}X^{{q^{m-i}}} =
\left(\sum_{i=0}^{t}a_{t-i}^{q^{m-t+i}}X^{q^{i}}\right) \circ
X^{q^{m-t}}.$$

Let $V$ be the leftmost factor of $\adj{Q}$ in the above
decomposition. It is a $q$--polynomial of $q$--degree $t$, and
$E\circ\adj{Q}=0$ leads to $E\circ V\circ X^{q^{m-t}}=0$. Since
$X^{q^{m-t}}$ is invertible in $\qpoly{}/( X^{q^m} - X )$,
we get $E\circ V=0 \mod (X^{q^m}-X)$.
\end{proof}

The goal is to compute this right--hand side annihilator $V$. It
satisfies
\begin{equation}\label{eq:right_wb} Y\circ V = C\circ V + E\circ V \equiv
C\circ V \mod (X^{q^{m}}-X).
\end{equation}
Equation
\eqref{eq:right_wb} leads to a non linear system of $n$ equations
whose variables are the $t+k+1$ unknown coefficients of $C$ and $V$.

\begin{equation}\label{eq:right_wb_nl_system} \left\lbrace
	\begin{array}{l} (Y\circ V)(g_{i}) = C\circ V(g_{i}) \\ \qdeg V
\le t \\ \qdeg C \le k-1.
	\end{array} \right.
\end{equation}

\noindent Due to the non linearity, it is not clear how this can efficiently be
solved. That is why we consider instead the following linearised
system
\begin{equation}\label{eq:right_wb_sl_system} \left\lbrace
	\begin{array}{l} (Y\circ V)(g_{i}) = N(g_{i}) \\ \qdeg V \le t \\
\qdeg N \le k + t - 1,
	\end{array} \right.
\end{equation}
whose unknowns are the
$k+2t+1$ coefficients of $N$ and $V$. The latter is \emph{a priori}
more general than the former. But we can link the set of solutions of
the two systems. This is specified in the following two propositions.

\begin{proposition} Any solution $(V, C)$ of \eqref{eq:right_wb_nl_system}
gives a solution $(V, N=C\circ V)$ of \eqref{eq:right_wb_sl_system}.
\end{proposition}

\begin{proof}
This is the direct analogue of \cite[Proposition~1]{L06a}.
\end{proof}

\begin{proposition} Assume that $E$ is of rank $t\le \lfloor
\frac{n-k}{2}\rfloor$. If $(V, N)$ is a nonzero solution of
\eqref{eq:right_wb_sl_system} then $N = C\circ V$ where $C = Y-E$ is
the interpolating $q$--polynomial of the codeword.
\end{proposition}

\begin{proof} Let $(V, N) \neq (0, 0)$ be a solution of
\eqref{eq:right_wb_sl_system}, and let $C$ be the $q$--polynomial of
$q$--degree strictly less than $k$ that interpolates the codeword. Let
$R \eqdef N - C\circ V$. It is a $q$--polynomial, of $q$--degree at
most $k-1+t$. Assume that $R \neq 0$. Then,
\[
  (Y-C)\circ V = Y\circ V - C\circ V = N - C\circ V \equiv R \mod (X^{q^m}-X)
\]
  \noindent \ie
  \begin{equation}\label{eq:proof_eq_systems} E\circ V \equiv R \mod (X^{q^m}-X).
  \end{equation}
  \noindent Hence, $\rk(R) \le \rk(E) \le t$. Since $R\neq 0$,
  $\qdeg R \ge \dim\ker R$. Therefore, by the rank--nullity theorem,
  \[ n = \dim\ker R + \rk(R) \le \qdeg R + \rk(R) \le k-1+2t \le n-1 <
n
  \] which is a contradiction. Therefore, $R$ must be zero, \ie $N = C\circ V$.
\end{proof}

Thenceforth, whenever $t\le\lfloor\frac{n-k}{2}\rfloor$, any non zero
solution of \eqref{eq:right_wb_sl_system} allows to recover the
codeword by simply computing a right--hand side Euclidean division, which can be
done efficiently (see \cite{O33}). The decoding process boils down to
solving the system of equations
\eqref{eq:right_wb_sl_system}. However, despite the transformation,
the system is only semi-linear over $\Fqm$. To address this issue, we will again use
the adjoint of a (class of) $q$--polynomial. Let
$\wadj{y_{i}}\eqdef \adj{Y}(g_{i})$,
for all $i=1,\dots, n$. Using the anticommutativity of the adjoint
operator, system \eqref{eq:right_wb_sl_system} is equivalent to

\begin{equation}\label{eq:right_wb_linear_system}
\adj{V}(\wadj{y_{i}}) = \adj{N}(g_{i}) \mbox{ for } i=1,\dots,n.
\end{equation}

\noindent which is now an $\Fqm$--linear system of $n$ equations
whose unknowns are the coefficients of $\adj{V}$ and $\adj{N}$, that
are in explicit one-to-one correspondence with the coefficients of $V$
and $N$.

\begin{algorithm}\label{algo:right_welch_berlekamp}
	\caption{Right--hand side variant of Welch--Berlekamp}
\DontPrintSemicolon \KwIn{$q$ a prime power, $k,n,m$ integers,
$\bfg=(g_{1},\dots,g_{n})$ a basis of $\Fqm/\Fq$, $\CC$ a Gabidulin
code of dimension $k$ and support $\bfg$, $t\le
\lfloor\frac{n-k}{2}\rfloor$ an integer, $\bfy\in\Fqm^{n}$.}
\KwOut{$\bfc\in\CC$ such that $\bfy=\bfc+\bfe$ for some
$\bfe\in\Fqm^{n}$ with $\rk(\bfe)\le t$.}  Find $Y$ the
$q$--polynomial of $q$--degree strictly less than $n$ such that
$Y(g_{i}) = y_{i}$\; Compute $\adj{Y}$ and evaluate in $\bfg$ to get
$\wadj{\bfy} \eqdef \adj{Y}(\bfg) \in\Fqm^{n}$\; Find a non zero solution
$(V_{0}, N_{0})$ of the linear system
\eqref{eq:right_wb_linear_system}\; Compute $V \eqdef \adj{V_{0}}$ and $N
\eqdef \adj{N_{0}}$\; Recover $C$ by computing the right--hand side
Euclidean division of $N$ by $V$\; \KwRet{$\bfc \eqdef C(\bfg)$}
\end{algorithm}

An implementation of this algorithm using SageMath v9.2 \cite{sage9.2} can be found on Github: \url{https://github.com/mbombar/Attack_on_LIGA}.

\if\longversion1
\begin{remark}
  This right--hand side algorithm can be generalised in order to
  decode an $u$--interleaved Gabidulin code (see for instance
  \cite{WZ14} for further reference about interleaved Gabidulin codes
  and their decoding algorithms). Indeed, let $\bfY = \bfC + \bfE$ be
  a corrupted codeword of an $u$--interleaved Gabidulin code, with
  $\bfE\in\Fqm^{u\times n}$ being an error matrix of $\Fq$-rank equal
  to $t$. Then, the rows of $\bfE$ as seen as vectors of $\Fqm^{n}$
  share a common \emph{row} support, namely the \emph{row} support of
  $\bfE$, of dimension $t$. Hence, they also share a common right
  annihilator, of $q$--degree at most $t$. The algorithm from \cite[\S
  4]{LO06} where the errors shared a common \emph{column} support can
  be adapted straightforwardly in this setting, and allows to decode
  almost all error matrix $\bfE$ of rank-weight
  $t\le \lfloor\frac{u}{u+1}(n-k)\rfloor$. This can be used in order
  to attack the original Faure--Loidreau cryptosystem in the same
  fashion as \cite[\S 3]{RPW20}.
\end{remark}
\fi

\fi

\bibliographystyle{splncs04}

\end{document}